\documentclass{amsart}
\usepackage{soul}
\usepackage[title]{appendix}
\usepackage{multirow}
\usepackage{graphicx}
\usepackage{amsfonts}
\usepackage{url}
\usepackage{amscd}
\usepackage{amssymb}
\usepackage{amsmath}
\usepackage{longtable}
\usepackage{algorithm}
\usepackage[noend]{algpseudocode}
\usepackage[table]{xcolor}  

\usepackage[hidelinks]{hyperref}

\newtheorem{remark}{Remark}[section]
\newtheorem*{theorem*}{Theorem}
\newtheorem{theorem}{Theorem}[section]
\newtheorem{lemma}[theorem]{Lemma}	

\newtheorem{definition}[theorem]{Definition}
\newtheorem{proposition}[theorem]{Proposition}
\numberwithin{equation}{section}

\usepackage{indentfirst}
\usepackage[utf8]{inputenc}
\usepackage{color,amsmath,amssymb,graphicx}
\usepackage{tikz-cd}
\usepackage{amsmath,amssymb}
\usepackage{enumerate}
\usepackage{dsfont}
\usepackage{amsthm}
\usepackage{float}
\usepackage{mathtools}

\usepackage{listings}
\definecolor{codegreen}{rgb}{0,0.6,0}
\definecolor{codegray}{rgb}{0.5,0.5,0.5}
\definecolor{codepurple}{rgb}{0.58,0,0.82}
\lstdefinestyle{mystyle}{
    backgroundcolor=\color{white},   
    commentstyle=\color{codegreen},
    keywordstyle=\color{magenta},
    numberstyle=\tiny\color{codegray},
    stringstyle=\color{codepurple},
    basicstyle=\ttfamily\footnotesize,
    breakatwhitespace=false,         
    breaklines=true,                 
    captionpos=b,                    
    keepspaces=true,                 
    numbers=left,                    
    numbersep=5pt,                  
    showspaces=false,                
    showstringspaces=false,
    showtabs=false,                  
    tabsize=2,
    upquote=true,
    columns=fullflexible
}

\author[E. Poimenidou]{Eirini Poimenidou}
\address{E. Poimenidou, School of Informatics, Aristotle University of Thessaloniki, Greece}
\email{epoimeni@csd.auth.gr}
\author[K. A. Draziotis]{K. A. Draziotis}
\address{K. A. Draziotis, School of Informatics, Aristotle University of Thessaloniki, Greece}
\email{drazioti@csd.auth.gr}

\begin{document}
\title{Message Recovery Attack in NTRU via Knapsack}
\keywords{Public Key Cryptography; NTRU Cryptosystem; Lattices; Lattice Reduction; Modular Knapsack Problem; Shortest Vector Problem;}
\subjclass[2020]{94A60}
\maketitle
\begin{abstract}
  In the present paper, we introduce a message-recovery attack based on the Modular Knapsack Problem, applicable to all variants of the NTRU-HPS cryptosystem. 
Assuming that a fraction $\epsilon$ of the coefficients of the message ${\bf{m}}\in\{-1,0,1\}^N$ and of the nonce vector ${\bf r}\in\{-1,0,1\}^N$ are known in advance at random positions, 
we reduce message decryption to finding a short vector in a lattice that encodes an instance of a  modular knapsack system. 
This allows us to address a key question: how much information about ${\bf m}$, or about the pair $({\bf m},{\bf r})$, is required before recovery becomes feasible? 
A FLATTER reduction successfully recovers the message, in practice when $\epsilon\approx 0.45$. 
Our implementation finds ${\bf m}$ within a few minutes on a commodity desktop.
 \end{abstract}

\section{Introduction}
In 1996, Hoffstein, Pipher, and Silverman developed the NTRU cryptosystem, aiming to create robust encryption and signature systems, as detailed in \cite{hoffstein}.
 Its security is based on the difficulty of solving a system of linear equations over polynomial rings, a problem that is expected to remain hard even with quantum computers. NTRU cryptosystem has withstood over 25 years of cryptanalysis, and variants of it NTRU have been shown to be closely related to the Ring Learning With Errors (R-LWE) problem, whose hardness is supported by worst-case reductions on ideal lattices. NTRU is known for its exceptional performance and moderate key-size, making it a popular choice for embedded cryptography. It has been standardized by IEEE, X9.98, and PQCRYPTO, and was a finalist in the NIST post-quantum cryptography standardization effort. 

In the present work, we outline a message recovery attack on NTRU based on the Shortest Vector Problem (SVP). Our method assumes partial knowledge of the message ${\bf m}\in\{-1,0,1\}^N$ and/or of the nonce vector ${\bf r}\in\{-1,0,1\}^N$. Such leakage assumptions are standard in the literature: for example, several works on DSA \cite{dsa2,marios,Howgrave} recover secret keys by exploiting partial information about ephemeral keys.  In \cite{may-hints}, the authors use "hints" from an oracle to recover the secret key in Kyber. Also, Coppersmith type attacks assume some knowledge of the one prime in order to compute the remaining part. 

\subsection{Previous Work}
The NTRU cryptosystem was first subjected to a lattice-based attack in 1997 by Coppersmith and Shamir \cite{CopSha97}.  
Later, Gama and Nguyen \cite{Gama-Nguyen} exploited decryption failures to recover the secret key, under the assumption of access to a decryption oracle. Subsequently, in 2001, Gentry \cite{gentry} proposed an attack that is particularly effective when the parameter $N$ is composite.

A number of other approaches have been developed to attack NTRU. One notable line of work reformulates the NTRU problem as a system of multivariate quadratic equations over the binary field, utilizing Witt vectors \cite{gerald,witt}. Odlyzko \cite{Odlyzko} proposed a meet-in-the-middle strategy that partitions the search space into two halves, reducing time complexity at the cost of substantial memory usage. Building on this idea, Howgrave-Graham introduced a hybrid attack \cite{hybrid attack} that combines lattice reduction with a meet-in-the-middle approach. This hybrid technique has since become a standard method for assessing the security of lattice-based encryption schemes.

For NTRU variants operating under larger moduli than those used in the NTRUEncrypt standard, Albrecht, Bai, and Ducas \cite{Albrecht}, as well as Cheon, Jeong, and Lee \cite{Cheon}, independently extended and refined these hybrid and meet-in-the-middle techniques to remain effective in more demanding parameter regimes. Nguyen \cite{Nguyen - Boosting the hybrid} later enhanced these methods by clarifying their structure and further improving their efficiency. While the subfield attack proposed in the previous work surpasses several earlier strategies, it still falls short of the most advanced hybrid attacks in terms of performance.

More recently, in 2023, May and Nowakowski \cite{may_recent} introduced a powerful new attack on the latest NTRU encryption scheme. Their approach employs a carefully designed lattice, leveraging the BKZ algorithm in conjunction with the sieving methods from the G6K library. Transitioning away from the traditional Coppersmith-Shamir lattice, they construct a lattice based on the cyclotomic ring, achieving significant performance gains. 

Finally, in 2025, two message recovery attacks \cite{adam_draz,nutmic} based on Babai nearest plane algorithm are presented. The authors of \cite{ntru_vfk_attack} also implement a message recovery attack, after reducing the NTRU-lattice to a Voronoi First Kind (VFK) lattice and then use a polynomial exact CVP-algorithm to recover the message. 

\subsection{Our contribution}
 In this paper we aim to {recover the unknown message}. 
 We consider an adversary who observes ciphertexts of the NTRU-HPS encryption scheme and who is assumed to know $k$ coefficients of the ternary plaintext polynomial $m(x)\in\{-1,0,1\}^N$ or $k_1$ coefficients of $m(x)$ and $k_2$ coefficients of the unknown nonce $r(x)\in\{-1,0,1\}^N$. More precisely, knowledge of approximately $45\%$ of the coefficients of the pair of unknown polynomials $(m(x), r(x))$, suffices to reconstruct $r(x)$ and, consequently, recover the plaintext $m(x)$. Such partial information may be leaked via a side-channel, arise from protocol redundancy, or be exposed by format markers.

Our method differs from previous approaches~\cite{adam_draz,ntru_vfk_attack,nutmic}, 
as it relies on the \emph{Shortest Vector Problem (SVP)} rather than the 
\emph{Closest Vector Problem (CVP)}. The SVP is better understood, both 
theoretically and algorithmically, providing a more solid foundation for our analysis.
Our assumptions in the present work are clearer, as it involves only the message polynomial \(m(x)\) and the nonce $r(x)$. Moreover, the underlying problem we attack is different. 

In our attack we solve a modular knapsack with coefficients in \(\{-1,0,1\}\) (instead of \(\{0,1\}\)). In more details our method reduces to the following problem: 
Let a matrix $A \in \mathbb{Z}_q^{n \times m}$ and a vector $B \in \mathbb{Z}_q^{n \times 1}$, we seek a vector $X$ with entries in $\{-1,0,1\},$ 
$
\text{such that}\ A\,X = B\ \text{in}\  \mathbb{Z}_q.
$
We call this \textit{modular knapsack problem}, also known as  {\it{Inhomogeneous Short Integer Solution problem}} ($ISIS_q$ \cite{ISIS23}). Furthermore, the present paper contributes to the study of the previous problem both in theory and in practice. 

Finally, our attack can also be applied to the NTRU-HRSS scheme, where only its parameters would need to change. For NTRU-Prime our attack could also be effective, however more research is required since it employs a non cyclic linear system which is different from the system in our case.

\subsection{Plausibility of Partial-Plaintext Knowledge}
In a Key Encapsulation Mechanism (KEM) the message ${\bf m}$ is ideally sampled uniformly at random, however if an implementation instead derives ${\bf m}$ (or the Key Derivation Function (KDF) preimage) from a structured seed, large portions of its content may be predictable. Additionally, if a weak pseudorandom number generator (PRNG) is used, then an adversary who collects sufficiently many $(\mathbf{m}_i, \mathbf{c}_i)$ pairs can detect statistical biases in certain coefficients of ${m}_i$ and exploit them to recover partial information.

In~\cite{sidechannel} the authors present a side-channel attack that exploits leakage from \emph{schoolbook} (product-scanning) polynomial multiplication when one operand is small with coefficients in $\{-1,0,1\}$.

 In~\cite{sidechannel2} they propose a power-based side-channel attack targeting the random generation of polynomials in NTRU. By combining a chosen-plaintext message attack with the collection of numerous $(\bf{m},\bf{c})$ pairs, one can detect statistical biases in the distribution of the nonce $r(x)$. 
\subsection{Roadmap} In Section \ref{sec:Background} we provide information on lattices, on the NTRU cryptosystem, we introduce   modular knapsack problem and FLATTER reduction algorithm. Next, in Section \ref{sec:system} we provide some preliminaries about our attack, such as the construction of a system. In Section \ref{section:attack} we describe our attack, the experiments we conducted and the results we yielded. Finally, in Section  \ref{sec:conclusion} we provide a conclusion. Our work's corresponding implementation can be found on Github\footnote{\url{https://github.com/poimenidou/knapsack-message-recovery-attack}}.

\section{Background}\label{sec:Background}
\subsection{Lattices}

In this section we recall some well-known facts about lattices. Let ${{\bf{b}}_1,{\bf{b}}_2,\ldots,{\bf{b}}_n}$ be linearly independent vectors of ${\mathbb{R}}^{m}$.
	The set 
	\[\mathcal{L} = \bigg{\{} \sum_{j=1}^{n}\alpha_j{\bf{b}}_j :
	\alpha_j\in\mathbb{Z}, 1\leq j\leq n\bigg{\}}\]
	is called  a {\em lattice} and 
	the finite vector set $\mathcal{B} = \{{\bf{b}}_1,\ldots,{\bf{b}}_n\}$ is called a basis of 
	the lattice $\mathcal{L}$. 
	All the bases of $\mathcal{L}$ have the same number of elements, i.e. in our case $n,$ which is called
	{\em dimension} or {\em rank} of $\mathcal{L}$. If $n=m$, then	the lattice $\mathcal{L}$ is said to have {\em full rank}. 
	Let $M$ be the $n\times m$ matrix, having as rows the vectors 
	${\bf{b}}_1,\ldots,{\bf{b}}_n$. 
	If $\mathcal{L}$ has full rank, then the {\em volume} of the lattice
	$\mathcal{L}$ is defined to be the positive number
	$|\det{M}|.$  The volume, as well as the rank, are independent of the basis $\mathcal{B}$. It is denoted 
	by $vol(\mathcal{L})$ or $\det{\mathcal{L}}.$ 
	Let now ${\bf v}\in \mathbb{R}^m$, then $\|{\bf v}\|$ denotes the Euclidean norm of ${\bf v}$.  Additionally,  we denote by  $\lambda_1(\mathcal{L})$ the least of the lengths of vectors of 	$ \mathcal{L}-\{ {\bf 0} \}$. Finally, if ${\bf t}\in {\rm{span}}({\bf b}_1,...,{\bf b}_n)$, then by $dist(\mathcal{L},{\bf t}),$ we denote $\min\{\|{\bf v}-{\bf t}\|: {\bf v}\in \mathcal{L} \}$.

There are two main fundamental problems on lattices the Shortest Vector Problem (SVP) and the Closest Vector Problem (CVP).

\textbf{The Shortest Vector Problem} (SVP): Given a lattice $\mathcal{L}$ find a non zero vector $\bf{b} \in \mathcal{L}$ that minimizes the (Euclidean) norm $\|\bf b\|$.

\textbf{The Closest Vector Problem} (CVP): Given a lattice $\mathcal{L}$ and a vector   ${\bf t} \in {\rm{span}}({\bf b}_1,\dots,{\bf b}_n)$ that is not in $\mathcal{L}$, find a  vector $\bf{b} \in \mathcal{L}$ that minimizes the distance $\|{\bf b}-{\bf t}\|$.

\textbf{The approximate Shortest Vector Problem} (appr${\rm SVP}$): Given a lattice $\mathcal{L}$ and a function $f(n)$, find a non-zero vector ${\bf b}\in \mathcal{L},$ such that:\[ \|{\bf b}\|\leq f(n) \lambda_1(\mathcal{L}).\] Each choice of the function  $f(n)$ gives a different approximation of the Shortest Vector Problem. 

\textbf{The approximate Closest Vector Problem} (appr${\rm CVP}$): Given a lattice $\mathcal{L},$ a vector  ${\bf t} \in {\rm{span}}({\bf b}_1,\dots,{\bf b}_n)$  and a function $f(n)$, find a vector ${\bf b}\in \mathcal{L}$ such that, \[ \|{\bf b}-{\bf t}\| \leq f(n) dist(\mathcal{L},{\bf t}).\] 

\subsection{Lattice Basis Reduction}
The well known LLL algorithm \cite{lll}, solves SVP rather well in small dimensions but performs poorly in large dimensions. The inability of LLL and other lattice reduction algorithms to effectively solve apprSVP and apprCVP determines the security of lattice-based cryptosystems. We provide the definition of LLL reduced basis of a lattice ${\mathcal{L}}.$
\begin{definition}
A basis $\mathcal{B} = \{{\bf{b}}_1,\ldots,{\bf{b}}_n\}$ of a lattice $\mathcal{L}$ is called LLL-reduced if it satisfies the following conditions:
\\
\texttt{1.} $|\mu_{i,j}| = \frac{|{\bf b}_i \cdot {\bf b}^*_j|}{\|{\bf b}^*_j\|^2} \le \frac{1}{2} $
for every $i,j$ with $1\leq j < i \leq n$,
\\
\texttt{2.} $\|{\bf b}^*_i\|^2  \geq (\frac{3}{4}- \mu_{i,i-1}^2)\|{\bf b}^*_{i-1}\|^2$ for every $i$   with $1 < i \leq n$. 
\end{definition}

\begin{proposition}
Let $\mathcal{L}$ be a lattice of rank $n$. For every  LLL-reduced  basis $\mathcal{B} = \{{\bf{b}}_1,\ldots,{\bf{b}}_n\}$ of a lattice $\mathcal{L}$ we get,
$$\|{\bf b}_1\| \leq 2^{(n-1)/2}\lambda_1(\mathcal{L}).$$ 
Thus, an LLL-reduced basis solves the approximate SVP to within a factor of $2^{(n-1)/2}$.
\end{proposition}
For details on the algorithm you can refer to \cite[Proposition 1.11]{lll}. Finally, we need the following Lemma.
\begin{lemma}\label{Lemma:lll-bound}
Let ${\bf b}_1,{\bf b}_2,\dots,{\bf b}_n$ be an LLL-reduced basis of the lattice ${\mathcal{L}}\subseteq\mathbb{R}^m$, and let $\{{\bf x}_1,{\bf x}_2,\dots,{\bf x}_t\}\subseteq {\mathcal{L}}$ be linearly independent vectors in ${\mathbb{R}}^m$.  Then, for all $1\le j\le t$, we have:
\[
  \|{\bf b}_j\|^2 \;\le\; 2^{\,n-1}\;\max\bigl\{\|{\bf x}_1\|^2,\|{\bf x}_2\|^2,\dots,\|{\bf x}_t\|^2\bigr\}.
\]
\end{lemma}
For a proof see \cite[Proposition 1.12]{lll}.

\subsubsection{FLATTER Reduction}
FLATTER is a fast lattice reduction algorithm for integer lattice bases created by Keegan Ryan and Nadia Heninger \cite{flatter} in 2023. It enhances the classical LLL-style reduction through a recursive QR decomposition combined with precision compression at each recursion level. The algorithm provides guarantees similar to traditional LLL but with significantly better performance in practice. 

The authors define a new notion of reduced-ness, called $\alpha$-lattice-reduced, based on a metric they call the drop of a lattice basis. They define the drop of $B$ as the total amount of downward steps in the lattice profile $\ell_i$ with $\ell_i = \log{\|{\bf b}_i^*\|},$ where ($\mathbf{b}_i^*$ is the $i$-th Gram-Schmidt vector) for $i\in\{1,\dots,n$\}:
$$drop(B)=vol\Biggl( \bigcup_{\substack{1 \leq i \leq n-1, \ell_{i+1} < \ell_i}}[\ell_{i+1}, \ell_i] \Biggr) = \sum_{i:\ell_{i+1}<\ell_i}(\ell_i-\ell_{i+1}).$$

A basis $B$ of rank $n$ is $\alpha$-lattice-reduced if it is \textit{size-reduced}, meaning that the upper-triangular Gram-Schmidt matrix has bounded coefficients, and if the drop of $B$ is smaller or equal to $\alpha n$, i.e. $drop(B) \leq \alpha n$. According to Theorem 2 of \cite{flatter}, if $B$ is $\alpha$-lattice-reduced, then it satisfies analogous bounds to LLL:
\begin{align*}
\|\mathbf{b}_1\| &\le 2^{\alpha n} (\det B)^{1/n} \\
\|\mathbf{b}_n^*\| &\ge 2^{-\alpha n} (\det B)^{1/n} \\
\|\mathbf{b}_i\| &\le 2^{\alpha n + O(n)}\, \lambda_i(B), \quad
    \text{for all } i \in \{1, \dots, n\}, \\
\prod_{i=1}^n \|\mathbf{b}_i\|
    &\le 2^{\alpha n^2 + O(n^2)} \det B,
\end{align*}
where $B = [\mathbf{b}_1, \dots, \mathbf{b}_n]$ and $\lambda_i(B)$ is the $i$-th successive minimum of the lattice generated by $B$. Hence, $\alpha$ plays the same mathematical role as $\delta-$constant in the Lov\'asz inequality that governs how orthogonal and how computationally expensive the reduced basis will be. Formally, a Flatter-reduced basis is not LLL-reduced because it replaces the Lov\'asz condition with a new global drop condition, but the resulting basis satisfies the same geometric and approximation guarantees as an LLL-reduced basis, effectively making it an equivalent form of lattice reduction.

The FLATTER algorithm has an asymptotic heuristic running time of
$$ O\!\left(n^{\omega}(C + n)^{1+\varepsilon}\right), $$
where $\omega \in (2,3]$ is the matrix multiplication exponent, $C = \log(\|B\|\|B^{-1}\|)$ (where $\|.\|$ denotes the spectral norm) bounds the condition number of the input basis and $\varepsilon > 0$ is an arbitrarily small constant accounting for the subpolynomial overhead of fast arithmetic operations on $O(C+n)$-bit numbers. This cost arises from recursive compression and sublattice-reduction steps, each dominated by size reduction and QR factorization at precision $O(C + n)$. The resulting complexity matches that of prior heuristic recursive methods \cite{improvedlll,segment,floatinglll} but is obtained under significantly weaker assumptions, allowing FLATTER to achieve LLL-equivalent reduction quality with practical, near--matrix-multiplication speed.

The implementation of this new algorithm was benchmarked extensively by the authors, against fpLLL\footnote{https://github.com/fplll/fplll} (the current gold-standard implementation) and previous recursive methods. Across a range of lattices, FLATTER outperformed existing tools and could successfully reduce even lattices of dimension 8192 in 6.4 core years. FLATTER's implementation is available on Github\footnote{https://github.com/keeganryan/flatter}. All the experiments in this paper failed to yield a correct result with the usage of fpLLL's reduction algorithms (we used a combination of LLL and BKZ with different blocksizes) while FLATTER's algorithm successfully reduced our input bases.

\subsection{NTRU-HPS}
In this section, we discuss about the NTRU-HPS. Let the polynomial ring ${\mathcal{R}}={\mathbb{Z}}[x]/\langle x^N-1 \rangle$ and we write $\star$ for the multiplication in the ring. If 
\[a(x)=a_{N-1}x^{N-1}+\cdots +a_0\ \text{ and }\ b(x)=b_{N-1}x^{N-1}+\cdots +b_0,\]
then 
$c(x)=a(x)\star b(x),$ is given by
$$c_k = \sum_{i+j\equiv k\ {\rm{mod}}{N}}a_ib_j,\ 0\le k\le N-1.$$

Alice selects public parameters $(N, q, d)$, with $N$ being prime number and $\gcd(q,N)$ $=\gcd(3,q)=1.$ Usually $N$ and $q$ are large, and $q$ is a power of $2.$	
With 
${\mathcal{T}}_a$ we denote the set of ternary polynomials\footnote{A ternary polynomial is one that has as coefficients only the integers $-1,0,1.$ } of ${\mathcal{R}}$ with degree at most $a$ and $\mathcal{T}_a(d_1,d_2)\subset {\mathcal{T}}_{a}$ consists from elements of ${\mathcal{T}}_{a}$ with $d_1$ coefficients equal to $1$ and $d_2$ equal to $-1.$ 

For her private key, Alice randomly selects $({f}(x),{g}(x))$ such that ${ f}(x) \in {\mathcal{M}}_f={\mathcal{T}}_{N-2}$ and ${ g}(x) \in {\mathcal{M}}_g={\mathcal{T}}_{N-2}(\frac{q}{16}-1,\frac{q}{16}-1)$. It is important that ${f}(x)$ is invertible in both ${\mathcal{R}}/q$ and ${\mathcal{R}}/3$. The inverses in ${\mathcal{R}}/3$ and ${\mathcal{R}}/q$ can be efficiently computed using the Euclidean algorithm and Hensel's Lemma, see \cite[Proposition 6.45]{hoffstein}. Let ${F}_q(x)$ and ${ F}_3(x)$ represent the inverses of ${ f}(x)$ in ${\mathcal{R}}/q$ and ${\mathcal{R}}/3$, respectively.	
	
 Alice next computes
\begin{equation}
{ h}(x) = {F}_q(x) \star { g}(x)\mod{q}.
\label{eq:pubkey}
\end{equation}
The polynomial ${h}(x)$ is Alice's public key.

Bob's plaintext is a polynomial ${ m}(x) \in {\mathcal{M}}_m,$ where
${\mathcal{M}}_m={\mathcal{T}}_{N-2}(\frac{q}{16}-1,\frac{q}{16}-1)$.  Then he chooses a random ephemeral key	${ r}(x) \in {\mathcal{M}}_r={\mathcal{T}}_{N-2}$ and computes the ciphertext,
\begin{equation}
{ c}(x) = 3{ r}(x) \star { h}(x) + { m}(x) \, \bmod\,  q.
\label{eq:encrypt}
\end{equation}
Finally, Bob sends to Alice the ciphertext
${ c}(x)\in {\mathcal{R}}/q.$
	
	To decrypt, Alice computes $${ v}(x) = {f}(x) \star {c}(x) \, \bmod\,  q.$$
Then, she centerlifts ${v}(x)$ to an element of ${\mathcal{R}},$ say ${v}'(x),$  and she finally computes,
	\[{b}(x) = {F}_3(x) \star {v}'(x) \, \bmod\,  3.\] 
Therefore, ${b}(x)$ is equal to the plaintext ${m}(x)$ (this is true when a simple inequality between $d, q,$ and $N$ is satisfied). For the exact values of $N, q$ we will use the ones proposed by NIST~\cite{nist2-ntru}.

\subsection{Knapsack Problem}
Here we discuss the knapsack problem in cryptography.
\begin{definition}[Knapsack Problem]
   Given $\textbf{a}=(a_1,\dots ,a_m)$, $ a_i \in \mathbb{N}$ and $s \in \mathbb{Z}$, find $\textbf{x}=(x_1,\dots ,x_m) \in \{0,1\}^m$ if it exists, such that $$\sum_{i=1}^{m} a_i x_i = s.$$
\end{definition}

This problem is NP-complete~\cite{np-complete}. 
A variation of the knapsack problem is the modular knapsack problem, which is the following:
\begin{definition}[Modular Knapsack Problem]
   Given a modulus $q$, vector $\textbf{a}=(a_1,\dots ,a_m)$, $ a_i \in \mathbb{N}$ and integer $s \in \mathbb{Z}$, find $\textbf{x}=(x_1,\dots ,x_m) \in \{0,1\}^m$ if it exists, such that $$ \sum_{i=1}^{m} a_i x_i \equiv s \mod q.$$
\end{definition}

We shall provide more general definitions using an abstract Abelian (additive) group $G.$
\begin{definition}($(G,m,{\mathcal{B}}$)-knapsack).
    For an Abelian group $G$ (written additively), integer $m$, and ${\mathcal{B}}\subset\mathbb{Z}$ small with $0\in {\mathcal{B}}$, 
    the \emph{$(G,m,{\mathcal{B}})$-knapsack problem over $G$,} asks: given elements $g_1,\dots,g_m\in G$ and a target $s\in G$, 
    find coefficients $x_i\in {\mathcal{B}}$ such that
    \[s =\sum_{i=1}^m x_i\,g_i.\]
\end{definition}

For instance, if $G={\mathbb{Z}}_q^{n}$ for some positive integer $q,$ and say $A\in G^{m}$ i.e. $A$ is a $n\times m$ matrix with columns in $G$ we are asking for solutions ${\bf x}\in {\mathcal{B}}^m$ such that $A{\bf x}^T={\bf s}^T\pmod{q}.$ It's worth noting that the knapsack problem over ${\mathbb{Z}}_q$ is equivalent to the Inhomogeneous Short Integer Solution problem $ISIS_q$~\cite{ISIS23}. 

There are three type of attacks in knapsack problems. Meet-in-the-middle, branch-and-bound and lattice-based. In meet-in-the-middle attacks~\cite{hybrid attack} the set of variables is split into two halves, all partial sums for each half are computed and then a collision between the two halves is searched for, that reconstructs the target sum. In branch-and-bound attacks~\cite{papadopoulou,branch-and-bound} depth-first search strategies are used to explore the solution space of integer combinations systematically, pruning subtrees that cannot lead to valid solutions and at the same time bounding the partial sum, norm of the solution or residual target distance. Finally, with lattice-based attacks one can get a solution to the knapsack problem by reducing the problem to the CVP or the SVP in certain lattices or by using reduction algorithms to find short vectors in a lattice that corresponds to valid knapsack solutions. 
The latter is what the authors in~\cite{lenstra-system} did in their lattice-based attack to the problem. We will call this method AHL and we describe it below.

\subsubsection{AHL Knapsack Algorithm}
Let $A$ be an integer $N \times k$ matrix, with $k\leq N$ and  $\textbf{s}$ an integer column $k$-vector. Aardal, Hurkens and A. Lenstra~\cite{lenstra-system} developed an algorithm to solve a system of Diophantine equations $AX= \textbf{s}$, with lower and upper bounds $0\leq X\leq \textbf{u}$, where $\textbf{u}$ is an integer $N$-vector. This is an NP-complete problem. In the absence of bound constraints it can be solved in polynomial time, for instance using Smith Normal Form (SNF)\footnote{See Appendix \ref{appendixA} for details about SNF.}. 

The authors first create the matrix $B$:

$$	B =\left[\begin{array}{c|c|c}

			I_N & \textbf{0}_{N\times 1} & N_2 A_{N\times k}  \\
			\hline
			\textbf{0}_{1\times N} & N_1  & -N_2 \textbf{s}_{1\times k} \\
            
		\end{array}\right],
$$
where $N_1,N_2$ are two positive integer numbers with $N_1 < N_2$. Then they use the LLL reduction algorithm to get the reduced form of the basis formed by the columns of $B$, which they denote $\hat{B}$. 
For suitably chosen $N_1, N_2 \in \mathbb{Z}$, the vector $\mathbf{x}$ is given by the first $N$ entries of the $(N - k + 1)$-th column of $\hat{B}$.
We shall use the previous ideas to implemented a ${\rm{mod}}\ q$ variant and we shall use it for our attack. In this paper, we work with row representations; that is, when a reduction is applied to a matrix, it refers to a row reduction.

\section{Construction of our Lattice}\label{sec:system}

 Starting from the NTRU-HPS encryption equation (\ref{eq:encrypt}) we get:
\begin{equation}\label{equation:ATTACK}
3^{-1} (c(x)-m(x)) \equiv h(x) \star r(x) \mod{q}.
\end{equation}
Our goal is to recover $r(x)$ (knowing $r(x)$ is equivalent to knowing $m(x)$). The idea is the following: using the information about $m_i$'s ($i=1,2,...,k)$ to construct a linear system in ${\mathbb{Z}}_q$ with $k-$equations and $N-$unknowns. We already know that it has  a small solution, namely the nonce ${\bf r}$. We shall apply lattice base methods, described earlier, to find this small solution.

The left-side of the previous equation (\ref{equation:ATTACK}) can be written as: 
\[
\sum_{j=0}^{N-1} 3^{-1}(c_j - m_j) x^j \in \mathcal{R}/q,
\]
and the right-side of (\ref{equation:ATTACK}) as:
\[
h(x) \star r(x) = \sum_{\ell =0}^{N-1} a_ \ell x^\ell \ \text{and}\ a_\ell = \sum_{i+j \equiv \ell \bmod{N}} h_jr_i.
\]
Furthermore, we define the vectors ${\bf a}_\ell \in {\mathbb{Z}}^N$, 
\begin{equation}
\textbf{a}_\ell = (h_{(\ell \bmod N)} ,\; h_{(\ell-1 \bmod N)},\; \dots ,h_{(\ell-(N-1) \bmod N)}),\ (0\leq \ell \leq N-1).
\label{equation:RIGHTSIDEVECTORS}
\end{equation}
Now, if we set ${\bf r}=(r_0,...,r_{N-1})$ we define $a_{\ell}={\bf a}_{\ell}\cdot {\bf r}.$ From the encryption equation \eqref{eq:encrypt}, if we know $k$ coefficients of the message, say
$\{
m_{i_0},m_{i_1},\dots,m_{i_{k-1}}
\},$
where
\[
\mathbf m=(m_0,m_1,\dots,m_{k-1},\dots,m_{N-1}),
\]
then we uniquely determine the integers
$\{
a_{i_0},a_{i_1},\dots, a_{i_{k-1}}
\},$ defined earlier.

Without loss of generality we assume that we know the first $k$ coefficients of $m(x).$ We form the $k \times N$ matrix $A$ with the vectors ${\bf a}_0,{\bf a}_1,...,{\bf a}_{k-1},$ as rows:
\begin{equation}\label{equation:MATRIXA}
A = 
\begin{bmatrix}
- & {\bf a}_0 & - \\
- & {\bf a}_1 & - \\
 & \cdots \\
- & {\bf a}_{k-1} & - \\
\end{bmatrix} =
\begin{bmatrix}
h_{0} & h_{N-1} & \dots & h_{1} \\
h_{1} & h_{0} & \dots & h_{2} \\
\vdots & \vdots & \ddots & \vdots \\
h_{k-1} & h_{k-2} & \dots & h_{(k-N){\rm{mod}} N}
\end{bmatrix}.
\end{equation}
Thus, we have $A{\bf r}^T={\bf T}_k,$ where 
\begin{equation}\label{Tk}
{\bf T}_k=(a_0,...,a_{k-1})^T.
\end{equation}
Our system (over ${\mathbb{Z}}_k$),
\begin{equation}\label{system}
A{\bf X}={\bf T}_k,
\end{equation}
will have $k$ equations and $N$ unknowns, ${\bf X}$ is a $N \times 1$ column vector that represents the unknown polynomial $r(x)$ and ${\bf T}_k$ is the $k \times 1$ column vector with the known entries $a_0,...,a_{k-1}$.

We construct the matrix $B_k$ of dimension $(N+k+1)\times (N+k+1)$,
\begin{equation}\label{relation:Bk}
 	B_k =\left[\begin{array}{c|c|c}

			I_N & \textbf{0}_{N\times 1} & N_2 A^T  \\
			\hline
			\textbf{0}_{1\times N} & N_1  & -N_2 \textbf{T}^T_{k} \\
            \hline
			\textbf{0}_{k\times N} & \textbf{0}_{k\times 1}  & N_2qI_k \\
            
		\end{array} \right] (A^T\ \text{is }{N\times k}),   
\end{equation}
where $A$ is defined in (\ref{equation:MATRIXA}), $N_1,N_2$ are positive integers which we shall determine later,  and $\mathcal{L} \subset\mathbb{Z}^{N+k+1}$ is the lattice generated by the rows of $B_k.$ This lattice is $q-$ary of full rank of dimension $N+k+1$ and volume $N_1(N_2q)^k$. 

The basis vectors are:
\begin{flalign*}
\textbf{b}_0  &= \bigl(\,\underbrace{1,0,\dots,0}_{\!N\text{ entries}\!},\,0\ ,
N_2h_0,\,N_2h_{1},\,\dots,\,N_2h_{k-1}\bigr) \\
\textbf{b}_1  &= (0, 1, \dots,0,\ 0\ , N_2h_{N-1}, N_2h_0, \dots, N_2 h_{k-2}) \\
\vdots & \\
\textbf{b}_{N-1}  &= (\underbrace{0, \dots, 0,1}_{\!N\text{ entries}\!},\,0\ , N_2h_{1}, N_2h_{2}, \dots, N_2 h_{(N-k){\rm{mod} N}}) \\
\textbf{b}_N  &= (\underbrace{0,0,\dots,0}_{\!N\text{ entries}\!},\  N_1\ , -N_2a_{0},  -N_2a_{1}, \dots,  -N_2a_{k-1}) \\
\textbf{b}_{N+1}  &= (\underbrace{0,0,\dots,0}_{\!N\text{ entries}\!},\ 0 \ ,\ \underbrace{N_2q,0, \dots, 0}_{\!k\text{ entries}\!}) \\
\vdots & \\
\textbf{b}_{N+k}  &= (\underbrace{0,0,\dots,0}_{\!(N+k)\text{ entries}\!}, N_2q). \\
\end{flalign*}
The lattice points are of the form,
$ (\lambda_0,\dots,\lambda_{N-1},N_1\lambda_N,N_2\beta_0,\dots,N_2\beta_{k-1} ),$
where $\lambda_j,\beta_j$ are integers (with $q|\beta_j$). 
In more details, let $(\lambda_0, \lambda_1, \dots,  \lambda_{N+k})$
 integer vector, and set ${\bf \Lambda}_N=(\lambda_0, \lambda_1, \dots,\lambda_{N-1}).$
Then the lattice points are of the form:

\begin{equation}
\label{lattice_points}
\begin{array}{rcl}
\displaystyle
\sum_{j=0}^{N+k}\lambda_j{\bf b}_j
&=&
\Big(
\lambda_0,\dots,\lambda_{N-1},\, N_1\lambda_N, \\[0.5em]
&&
\quad N_2\underbrace{({\bf \Lambda}_N\!\cdot\! \mathbf a_0 - \lambda_N a_0)}_{\text{1st equation of (\ref{system})}} + N_2q\lambda_{N+1}, \\[0.5em]
&&
\quad N_2\underbrace{({\bf \Lambda}_N\!\cdot\! \mathbf a_1 - \lambda_N a_1)}_{\text{2nd equation}} + N_2q\lambda_{N+2}, \\[0.5em]
&&
\quad \quad \quad \quad \quad \quad \quad  \quad \quad  \vdots \\[0.5em]
&&
\quad N_2\underbrace{({\bf \Lambda}_N\!\cdot\! \mathbf a_{k-1} - \lambda_N a_{k-1})}_{k\text{th equation}} + N_2q\lambda_{N+k}
\Big),
\end{array}
\end{equation}

\noindent
or equivalently,
\begin{equation*}
\begin{array}{rcl}
\displaystyle
\sum_{j=0}^{N+k}\lambda_j{\bf b}_j
&=&
\Big(
{\bf \Lambda}_N,\, N_1\lambda_N, \\[0.5em]
&&
\quad N_2({\bf \Lambda}_N\!\cdot\! \mathbf a_0 - \lambda_N a_0 + \lambda_{N+1}q), \\[0.5em]
&&
\quad N_2({\bf \Lambda}_N\!\cdot\! \mathbf a_1 - \lambda_N a_1 + \lambda_{N+2}q), \\[0.5em]
&&
\quad \quad \quad \quad \quad \quad \quad \quad  \vdots \\[0.5em]
&&
\quad N_2({\bf \Lambda}_N\!\cdot\! \mathbf a_{k-1} - \lambda_N a_{k-1} + \lambda_{N+k}q)
\Big).
\end{array}
\end{equation*}

If we can find a vector of the form:
\begin{equation}\label{lambda_N}
  ({\boldsymbol \Lambda}_N, N_1, 0, \ldots, 0) \in \mathbb{Z}_q 
\quad (\text{i.e., } \lambda_N = 1)  
\end{equation}
for some suitable integers 
\((\lambda_0, \ldots, \lambda_{N-1})\),
then \({\boldsymbol \Lambda}_N\) constitutes a solution to the system 
\[
A{\bf X} = {\bf T}_k \quad \text{over } \mathbb{Z}_q.
\]
Also the inverse is true. If there are $\lambda_0,\lambda_1,...,\lambda_{N-1}$ such that $A{\bf \Lambda}_N^T={\bf T}_k$, then the vector ${\bf z}=({\bf \Lambda}_N,N_1,qN_2\rho_{N+1},...,qN_2\rho_{N+k})$ belongs to ${\mathcal{L}}.$ Indeed, ${\bf z}$ is written as the integer linear combination: $\sum_{j=0}^{N-1}\lambda_j{\bf b}_j+N_1{\bf b}_{N+1}+\sum_{j=N+1}^{N+k}\rho_j{\bf b}_j.$ We proved the following.
\begin{lemma}\label{Lemma:shape_lemma}
The integer vector ${\bf x}=(x_0,\dots,x_{N-1})$ is a solution of $AX={\bf T}_k$ in ${\mathbb{Z}}_q$ if and only if $({\bf x}, N_1, qN_2\rho_{N+1},\dots,qN_2\rho_{N+k})$ is a point of the lattice ${\mathcal{L}}.$ 
\end{lemma}

\begin{remark}
Let ${\bf r}=(r_0,\dots,r_{N-1})$ be a nonce of NTRU. Then $A{\bf r}^T={\bf T}_k\pmod{q},$ so the vector 
$${\bf v}=(\pm{\bf r},\pm N_1,{\bf 0}_k)\in {\mathcal{L}},$$
where the signs are taken as $(+,+),(-,-).$ Indeed, by definition of ${\bf r}$ we have $A{\bf r}^T={\bf T}_k\pmod{q},$ so there is some vector ${\bf R}=(\rho_1,...,\rho_k)$ such that 
$A{\bf r}^T={\bf T}_k-q{\bf R}.$ Therefore, we choose $\lambda_N=\pm 1,\lambda_j=\lambda_N r_j \ (0\le j\le N-1)$, and $\lambda_{N+j}=\lambda_N\rho_j\ (1\le j\le k).$ Then, using $(\ref{lattice_points})$ we get ${\bf v}=\sum_{j=0}^{N+k}\lambda_j{\bf b}_j.$ 
Furthermore, the norm satisfies, 
\[\| {\bf v} \|^2=\|{\bf r} \|^2+N_1^2<N^2+N_1^2.\]
Thus, ${\bf v}$ is  a short vector of ${\mathcal{L}}.$ In general, small solutions of the system $AX={\bf T}_k\pmod{q}$ provide small non-zero vectors of ${\mathcal{L}},$ and the inverse.
\end{remark}
We shall prove that for suitable $N_1, N_2$ (and under some plausible conditions) the LLL-reduced matrix $\hat{B}_k$  is of the form
\begin{equation}\label{reduced_matrix}
  \hat{B}_k \;=\;
\begin{pmatrix}
\hat{b}_{0,0} & \cdots & \hat{b}_{0,\,N-1} & \varepsilon_{0} & 0 &\cdots & 0\\[6pt]
\vdots  & \ddots & \vdots      & \vdots            &        & \ddots            \\[6pt]
\hat{b}_{N-1,0} & \cdots & \hat{b}_{N-1,\,N-1} & \varepsilon_{N-1} & 0 & \cdots & 0 \\[6pt]
\hat{b}_{N,0}       & \cdots & \hat{b}_{N,N-1}          & \varepsilon_{N}  & *     & \cdots & *      \\[6pt]
\hat{b}_{N+1,0}       & \cdots & \hat{b}_{N+1,N-1}         & \varepsilon_{N+1}         & \multicolumn{3}{c}{\multirow{3}{*}{\Large${\bf C}_{k}$}} \\ 
\vdots  &        & \vdots      & \vdots        &        &        &          \\ 
\hat{b}_{N+k,0}        & \cdots & \hat{b}_{N+k,N-1}            & \varepsilon_{N+k}            &        &        &          
\end{pmatrix},
\end{equation}
where $\varepsilon_{i}\equiv 0\pmod {N_1}$ for some indexes $i$ and $\varepsilon_i=0$ for the remaining indexes, and ${\bf C}_k$ is a $k\times k$ matrix. We remark that, at least one of the ${\varepsilon}_i$ will be $N_1$ or $-N_1.$ 
To see this we first remark that LLL-reduction makes the following two
operations to the rows of $B_k,$
$${\text{row}}_j\leftrightarrow {\text{row}}_i$$
$${\text{row}}_j\leftarrow {\text{row}}_j-\lambda {\text{row}}_i\ (\lambda\in {\mathbb{Z}}),\ \text{for}\ j>i.$$ 
So, the $(N+1)$-column, i.e. the vector
$(0,...,0,N_1,0,...0),$ after the LLL-reduction shall contain only $0$ and some multiples of $N_1.$ Since the gcd of the $(N+1)$-column remains the same after the LLL reduction we get,
\begin{equation}\label{column_N+1}
\gcd(\varepsilon_0,...,\varepsilon_{N+k})=\gcd(0,...,0,N_1,0,...,0)=N_1.
\end{equation}
So, indeed $\varepsilon_{i}\equiv 0\pmod {N_1}.$ 

 Let a $k\times N$ $(k<N)$ matrix A with $rank(A)=k$. For the following Theorem we need the Smith Normal Form of $A$, which is given by:
\[
  D \;=\; P\,A\,Q,
  \quad P\in GL_k(\mathbb Z),\;Q\in GL_N(\mathbb Z),
  \quad D={\rm{diag}}(d_1,\dots,d_k,0,\dots,0)\in {\mathbb{Z}}^{k\times N}\]
with \(d_i\mid d_{i+1}\) and the last \(N-k\) diagonal entries of \(D\) are zero. Let,
\[ Q=[{\bf q}_1\mid\cdots\mid {\bf q}_N].\] Since, $rank(A)=k$, then  for each \(j\in \{k+1,\dots,N\}\) the columns ${\bf q}_j$ generate $Ker_{\mathbb{Z}}(A).$ 
In Appendix \ref{appendixA}, we provide the proof of the following Theorem.
\begin{theorem}\label{theorem}
If $V={\rm{span}}({\bf q}_{k+1},...,{\bf q}_{N})$, $W={\rm{span}}({\bf e}_1,...,{\bf e}_k)$ subspaces of ${\mathbb{R}}^N$ and $V\cap W=\{ {\bf 0} \},$ then there are $N_1$, $N_2$ such that, the LLL-reduced matrix of $B_k$ is of the form $\hat{B_k}.$ In fact we prove that 
$2^{N+k}N_1^2<c(N,k)<N_2^2$ for some constant $c(N,k).$
\end{theorem}
Say ${\hat{\bf{b}}}_{i_0}$ a row of the LLL reduced matrix ${\hat{B}}_k,$ that has $\varepsilon N_1$ ($\varepsilon\in \{-1,1\}$) in $N+1$ entry ($N+1$ entry is the element ${\hat{b}}_{i_0,N}$ since we started counting from $0$). Then,
the vector  
${\bf x}_{\varepsilon}=(\varepsilon \hat{b}_{i_0,0},\varepsilon \hat{b}_{i_0,1},\dots,\varepsilon\hat{b}_{i_0,N-1})$
is  a solution of the system $AX={\bf T}_k$ in ${\mathbb{Z}}_q$ if $i_0\le N-1.$  This is immediate from Lemma (\ref{Lemma:shape_lemma}). We shall use the previous idea to the attack presented below.

\section{The attack}\label{section:attack}

The procedure for the attack consists of two main stages. In the first stage, we construct the matrix $B_k$ (\ref{relation:Bk}) and apply the reduction routine \textsc{Flatter}. In the second stage, we extract from the $(n+1)$-th column of the reduced matrix (cf.\ (\ref{reduced_matrix})), the first $N$ entries, which we denote by
$(\varepsilon_0,\dots,\varepsilon_{N-1}).$ For every $\varepsilon_i$ we check whether $N_1 \mid \varepsilon_i$ and we define the quotient as $quotient = \varepsilon_i/N_1$. We also extract the first $N$ entries from the row indexed by $\varepsilon_i$ ($row=\hat B[index(\varepsilon_i)][N]$). If the quotient is an integer, then we form the possible solution as,
$$\mathbf r'=\bigl(row[0]/quotient,\dots,row[N-1]/quotient\bigr).$$ 
If the Euclidean norm $\|\mathbf r'\|$ is smaller than a prescribed threshold \texttt{app\_value}\footnote{The threshold is chosen from the expected norm of a random ternary vector.  If $r_i\in\{-1,0,1\}$ are i.i.d., then $S=\sum_{i=0}^{N-1} r_i^2\sim\mathrm{Bin}(N,2/3)$ so $E[S]=2N/3$. For large $N$ we have $E[\sqrt{S}]\approx\sqrt{E[S]}$, which yields the empirical thresholds \texttt{app\_value}$\approx 19$ for \texttt{ntruhps2048509} ($N=509$), $\approx 21$ for \texttt{ntruhps2048677} ($N=677$) and $\approx 24$ for \texttt{ntruhps4096821} ($N=811$).}, we return $\mathbf r'$, otherwise the algorithm reports failure. Note that the algorithm does not always recover the nonce in every instance.

In practice, we have observed that the solution vector can be found in any of the $[0,N-1]$ positions in the $(n+1)$-th column of the reduced matrix, contrary to what was stated in \cite{lenstra-system}, where they get the solution vector specifically from the $N+1$ row of the reduced matrix. 
In Algorithm \ref{alg:algo} we present the pseudocode of our attack and in Table \ref{table:attack1} the results of our experiments.

It is worth emphasizing that the reduction was carried out using FLATTER~\cite{flatter}. In contrast, the fpLLL implementation were unable to reproduce successful recoveries under identical parameters. 

\begin{algorithm}[H]
\caption{Message Recovery Attack}
\label{alg:algo}
\begin{algorithmic}[1]
\State \textbf{Input:} $N,N_1,N_2, \textbf{c}, \{m_0,m_1,\dots,m_{k-1}\}, app\_value$, ${\bf T}_k,q$
\State \textbf{Output:} the message \textbf{m} of the NTRU-HPS system or $null$
\State $\textbf{a} \gets [3^{-1}(c_i-m_i) \mod q \mid i \in [0,\dots,k-1]]$
\State $\textbf{a}_0 \gets [\textbf{H}_{k,i} \mid i \in [0,\dots k-1]]$ \Comment{ see subsection \ref{section:attack}}
\State $B \gets create\_basis(N,N_1,N_2,\textbf{a},\textbf{a}_0)$ \Comment{ see relation (\ref{relation:Bk})}
\State $\hat{B} \gets FLATTER(B)$ 
\State $column \gets \hat{B}^T[N]$ \Comment{get the $(N+1)$-th column of $\hat{B}$}
\For{$i \in \{0,\dots,N-1\}$}
    \State $quotient \gets column[i]/N_1$
    \If{$quotient \neq 0$}
        \State $row \gets \hat{B}[i][0 \ to \ N-1]$ \Comment{first $N$ elements of the $i$-th row}
        \If {$gcd(row) == |quotient|$}
        \State $\textbf{r}' \gets row/quotient$
                \If {$\|\textbf{r}'\| \leq app\_value $ and $\textbf{r}' \in \{-1,0,1\}^N$ and $A{\bf r}'\equiv{\bf T}_k\pmod{q}$} 
                                \State $m'(x) \gets (c(x)- 3 h(x)\star r'(x)) \mod q$
                                \State \Return centerlift($\textbf{m}'$)
            \EndIf
        \EndIf
    \EndIf
\EndFor
\State \Return $null$
\end{algorithmic}
\end{algorithm}

\begin{table}[h]
\centering
\begin{tabular}{|c|c|c|c|c|c|c|c|c|}
\hline
\textbf{($N, q$)} & $N_1$ & $x : N_2=\lceil q^{x}\rceil$ & $k$& $\%$ & runtime & rate \\
\hline\hline
(509, 2048) & 9 & 8 & 425 & $83\%$ & 5m & $100\%$ \\
\hline
(677, 2048) & 1 & 15 & 600 & $89\%$ & 12m & $90\%$  \\
\hline
(821, 4096) & 7 & 21 & 750 & $91\%$ & 17m & $50\%$  \\
\hline
\end{tabular}
\caption{Message recovery attack for  $N=509, \ N=677,$ and $ \ N=821.$ }
\label{table:attack1}
\end{table}

For all of our experiments, we choose $N_1$ to be much smaller than $N_2$, guided also by Theorem~\ref{theorem}. We ensure that the lattice vector \(\mathbf{v} = (\mathbf{r}', \pm N_1, \mathbf{0}_k)\) remains short, thereby improving the chances of recovering \(\mathbf{r}\) efficiently via lattice reduction. The choice of $N_1,\ N_2$ is crucial to the success of our algorithm. In Table \ref{table:attack1} we see the results of the message recovery attack where we present the different values of $N_1,\ N_2$ that we used for each dimension $N$, as well as the number of known coefficients of $m(x)$, notated as $k$. This value of $k$ is the optimal for each dimension $N$, i.e. for smaller values of $k$ we did not get any solution. The percentage column is calculated with $\frac{k}{N} * 100$ and to calculate the success rate, we have run $10$ experiments per row. The runtime column is the total (wall) time of the algorithm in an AMD Ryzen 7 3700X (16 cores) with 16 GB of RAM, Ubuntu machine. The code for the message recovery attack can be found on Github\footnote{\url{https://github.com/poimenidou/knapsack-message-recovery-attack/blob/main/attack.ipynb}}.

We presented an algorithm capable of finding the full message $m$ when approximately 80\% of its entries are known for $N=509$. This percentage can seem large enough but we can seemingly improve this. In the next part, we will discuss an alternative message recovery attack where the percentage of known elements of $m$ will be reduced by incorporating information of the auxiliary vector $r$.

\subsection{Alternative attack} 
In the present attack, we make the assumption that $k_1$ coefficients of the message $m(x)$ and $k_2$ coefficients of the nonce $r(x)$ are known by the attacker, and we show that we can recover the full nonce $r(x)$, and thus the full message. In the previous attack we had $k_2=0.$

The set $Z_0$ holds the positions of the zero-elements, the set $Z_1$ has the positions of the one-elements and $Z_{-1}$ has the positions of the minus-one-elements of the known part of $r(x)$. All these sets have $k_2$ elements in total.

We remove from the matrix $A$ the columns indexed by ${Z}_0$, $Z_1$, and $Z_{-1}$. Instead of solving the same system as in the previous attack $AX={\bf T}_{k_1}$ ($A$ has dimension $k_1\times (N-k_1))$ we will now solve $A_zX={\bf T}_{z}$, where $\textbf{T}_{z} = {\bf T}_{k_1}-{\bf S}$ where,
$${\bf S}=\sum_{i:r_i=1} {\rm{col}}_i(A)-\sum_{i:r_i=-1} {\rm{col}}_i(A).$$ 
Matrix $A_z$ is obtained from \(A\) by removing the columns indexed by ${Z}_0$, $Z_1$, and $Z_{-1}$. Now, the dimension of $A_z$  will be $k_1\times(N-k_2)$ and of $B_z$  will be $(N+k_1-k_2+1)\times (N+k_1-k_2+1)$, where $B_z$ is the matrix:
\begin{equation}
 	B_z =\left[\begin{array}{c|c|c}

			I_{N-k_2} & \textbf{0}_{(N-k_2)\times 1} & N_2 A_z^T  \\
			\hline
			\textbf{0}_{1\times (N-k_2)} & N_1  & -N_2 \textbf{T}^T_{z} \\
            \hline
			\textbf{0}_{k_1\times (N-k_2)} & \textbf{0}_{k_1\times 1}  & N_2qI_{k_1} \\
            
		\end{array} \right] (A_z^T\ \text{is }{(N-k_2)\times k_1}).   
\end{equation}

So, we first construct  the matrix $B_z$ and then, we use the FLATTER algorithm to get the reduced basis $\hat B_z$ and find the solution vector just like we did in the previous attack. 
After the reduction, we reconstruct $r(x)$ (by adding the known $k_2$ elements) and we check if it is valid i.e. its norm satisfies the upper bound given by the ${\it{app\_value}}$. 
\begin{algorithm}[H]
\caption{Alternative Message Recovery Attack}
\label{alg:alternativealgo}
\begin{algorithmic}[1]
\State \textbf{Input:} $N, N_1, N_2,  \textbf{c}, \{m_0,m_1,\dots,m_{k_1-1}\}, \{r_0,r_1,\dots,r_{k_2-1}\}, app\_value, {\bf T}_{k_1}, q$
\State \textbf{Output:} the message \textbf{m} of the NTRU-HPS system or $null$
\State $\textbf{a} \gets [3^{-1}(c_i-m_i) \mod q \mid i \in [0,\dots,k_1-1]]$
\State $\textbf{a}_0 \gets [\textbf{H}_{k_1,i} \mid i \in [0,\dots, k_1-1]]$ \Comment{ see subsection \ref{section:attack}}
\State $B_z \gets create\_basis(N,N_1,N_2,\textbf{a},\textbf{a}_0)$ \Comment{ see relation (\ref{relation:Bk})}
\State $B_z' \gets delete\_columns(B_z,N,N_1,N_2, k_1, \{\textbf{r}_i: i \in [0,\dots,k_2-1]) \}$ 
\State $N' \gets N-k_2$ 
\State $\hat{B_z} \gets FLATTER(B_z')$ 
\State $column \gets (\hat{B}_z^T)[N]$ \Comment{the $(N+1)$-th column of $\hat{B_z}$}
\For{$i \in \{0,\dots,N-1\}$}
    \If{$column[i] == N_1 \ \textbf{or} \ column[i] == -N_1$}
        \State $row \gets \hat{B_z}[i][0 \ to \ N'-1]$ \Comment{first $N'$ elements of the $i$-th row}
        \State $\textbf{r}' \gets set\_known\_positions(row, \{\textbf{r}_i: i \in [0,\dots,k_2-1]) \})$
        \If {$\|\textbf{r}'\| \leq app\_value $ \textbf{and} $\textbf{r}' \in \{-1,0,1\}^N$ \textbf{and} $A{\bf r}'\equiv{\bf T}_{k_1}\pmod{q}$} 
            \State $m'(x) \gets (c(x)- 3 h(x)\star r'(x)) \mod q$
            \State \Return centerlift($\textbf{m}'$)
        \EndIf
    \EndIf
\EndFor
\State \Return $null$
\end{algorithmic}
\end{algorithm}

For all the experiments we choose $N_1$ and $N_2$ as in the previous attack. In Table \ref{table:alteattack} we present the results of the alternative message recovery attack, with the different values of $N_1,\ N_2$ that we used for each dimension $N$, as well as the number of known elements of $m(x)$, notated as $k_1$ and the the number of known elements of $r(x)$, notated as $k_2$. The percentage column is calculated with $\frac{k_1+k_2}{2N} * 100$ and to calculate the success rate we have run $10$ experiments per row. The runtime refers to the total (wall) time our algorithm takes in an AMD Ryzen 7 3700X (16 cores) with 16 GB of RAM, Ubuntu machine. The highlighted rows are the ones with $\min(k_1+k_2)$ and the highest success rate of that specific dimension $N$. The code for the message recovery attack can be found on Github\footnote{\url{https://github.com/poimenidou/knapsack-message-recovery-attack/blob/main/attack_-101.ipynb}}.

\begin{longtable}{|c|c|c|c|c|c|c|c|c|}
\hline
\textbf{($N, q$)} & $N_1$ & $x : N_2=\lceil q^{x}\rceil$ & ${k_1}$ & ${k_2}$ & $k_1+k_2$ & $\textbf{\%}$ & runtime & rate \\
\hline\hline
(509, 2048) & 9 & 8 & 300  & 125 & 425 & $42\%$ & 3m & $10\%$\\
(509, 2048) & 9 & 8 & 250 & 185 & 435 & $43\%$ & 3m & $0\%$\\
\rowcolor{yellow!20}  
(509, 2048) & 9 & 8 & 300  & 135 & 435 & $43\%$ & 3m & $100\%$\\
(509, 2048) & 9 & 8 & 230 & 215 & 445 & $44\%$ & 2m & $50\%$\\
(509, 2048) & 9 & 8 & 250 & 195 & 445 & $44\%$ & 3m & $100\%$\\
(509, 2048) & 9 & 8 & 350 & 100 & 450 & $44\%$ & 3m & $100\%$\\
(509, 2048) & 9 & 8 & 230 & 225 & 455 & $45\%$ & 2m & $100\%$\\
\hline
(677, 2048) & 1 & 15 & 500 & 100 & 600 & $44\%$ & 10m & $10\%$ \\
(677, 2048) & 1 & 15 & 400 & 210 & 610 & $45\%$ & 5m & $0\%$ \\
\rowcolor{yellow!20}  
(677, 2048) & 1 & 15 & 500 & 110 & 610 & $45\%$ & 10m & $90\%$ \\
(677, 2048) & 1 & 15 & 400 & 215 & 615 & $45\%$ & 5m & $70\%$ \\
(677, 2048) & 1 & 15 & 400 & 220 & 620 & $46\%$ & 5m & $10\%$ \\
(677, 2048) & 1 & 15 & 315 & 310 & 625 & $46\%$ & 2m & $40\%$ \\
(677, 2048) & 1 & 15 & 315 & 315 & 630 & $47\%$ & 2m & $60\%$ \\
\hline
\rowcolor{yellow!20}  
(821, 4096) & 7 & 21 & 700 & 90 & 790 & $48\%$ & 14m & $90\%$ \\
(821, 4096) & 7 & 21 & 410 & 400 & 810 & $49\%$ & 2m & $0\%$ \\
(821, 4096) & 7 & 21 & 500 & 310 & 810 & $49\%$ & 3m & $0\%$ \\
(821, 4096) & 7 & 21 & 600 & 210 & 810 & $49\%$ & 5m & $10\%$ \\
(821, 4096) & 7 & 21 & 410 & 410 & 820 & $50\%$ & 2m & $100\%$ \\
(821, 4096) & 7 & 21 & 500 & 320 & 820 & $50\%$ & 3m & $100\%$ \\
(821, 4096) & 7 & 21 & 600 & 220 & 820 & $50\%$ & 5m & $100\%$ \\
\hline
\caption{Alternative attack for  $N=509,$ $ \ N=677,$ and  $ \ N=821$ }
\label{table:alteattack}
\end{longtable}

From Table \ref{table:alteattack} we observe that the most efficient attacks time-wise, require the value $k_1-k_2$ to be as minimal as possible, while preserving $k_1 \geq k_2$. The most efficient attacks are the ones with $k_1=k_2$, however they don't necessarily have the highest success probability when the objective is to minimize the sum $k_1+k_2$. In the end, the best combination of $k_1,\ k_2$ values depend on the information the attacker has, while always keeping in mind the previous observations.

This alternative attack cannot be compared directly to the previous attack since the initial assumptions differ. However, it stands as an improvement to the previous one since it cuts in half the total percentage of the known elements in the assumption about revealed entries ($40\%-45\%$ as opposed to $80\%-90\%$). It also offers greater flexibility since it permits asymmetric leakage between the plaintext $m$ and the nonce $r$, meaning that the attacker may know different numbers of entries of $m$ and $r$ provided that their combined amount meets the required threshold for the given dimension $N$. The alternative attack also decreases by a large percentage the total wall time of the reduction time and therefore of the algorithm.

To mitigate our attack, implementers should ensure that the plaintext $m$ is sampled uniformly at random rather than derived from structured or predictable protocol fields. Side-channel protections are essential, particularly for polynomial multiplication, which should be implemented in constant time.

\section{Conclusion}\label{sec:conclusion}
In the present paper, we propose a practical \emph{SVP}-based message-recovery attack on NTRU-HPS under partial plaintext leakage. 
Assuming that a subset of the coefficients of both the message and the nonce vector are known, we construct a linear system of modular equations whose underlying structure corresponds to a modular knapsack problem. 
We then reduce this problem to an instance of the \emph{Shortest Vector Problem (SVP)} on a suitably defined lattice. 
By applying the \textsc{FLATTER} lattice reduction algorithm, we are able to solve the modular knapsack instance and successfully recover the unknown polynomial $r(x)$, and consequently, the message $m(x)$.

Knowing the $42-48\%$ of the message and the nonce, we can attack the three variants of NTRU-HPS and recover the message in minutes on commodity hardware. For future work, we could address extensions of our method to other lattice-based schemes, such as Kyber and Saber.

\ \\\\
{\bf Acknowledgment}. {The second author was co funded by SECUR-EU. The SECUR-EU project funded under Grant Agreement 101128029 is supported by the European Cybersecurity Competence Centre.}

\begin{appendix}

\section{Proof of the Theorem \ref{theorem}}\label{appendixA}
We start with the following Theorem.
\begin{theorem}(Smith, 1861 \cite{Smith}).
Let $A\in {\mathcal{M}}_{m\times n}({\mathbb{Z}})$ of rank $r.$ Then there is a diagonal integer matrix $D={\rm diag}(\lambda_1,\lambda_2,...,\lambda_r,0,...,0)$ ($m\times n$) with 
$$ \lambda_1|\lambda_2|\cdots |\lambda_{r}$$
and unimodular matrices $U\in GL_m({\mathbb{Z}})$ and $V\in GL_n({\mathbb{Z}})$ such that 
$$A=UDV.$$
The non zero diagonal elements $\lambda_1,...,\lambda_r$ of $D$ are called {\it{elementary divisors}} of $A$ and are defined up to sign. $D$ is the Smith Normal Form (SNF) of $A.$
\end{theorem}

The first algorithms for computing the Hermite Normal Form (HNF) and Smith Normal Form (SNF) appeared in 1971 \cite{bradley}. Later, von zur Gathen and Sieveking \cite{von} presented improved algorithms, in 1976 that achieved polynomial-time complexity. For $A\in {\mathcal{M}}_{m\times n}({\mathbb{Z}})$ there are unimodular matrices $P,Q$ such that 
\[
PAQ=\left[\begin{array}{cc}
   {\rm{diag}}(\lambda_1,...,\lambda_r)    &  {\bf 0}_{r\times (n-r)}  \\
   {\bf 0}_{(m-r)\times r} & {\bf 0}_{(m-r)\times (n-r)}
		\end{array}\right],
   \]
		where $r$ is the rank of $A$ and $\lambda_i\in{\mathbb{Z}}_{>0},\ \lambda_i|\lambda_{i+1}.$
		Storjohann \cite[Theorem 12]{storj} provided a deterministic algortithm for computing SNF with complexity,
		$$O(nmr^2\log^{2}(r\|A\|)+r^4\log^{3}(\|A\|))\ \text{bit\ operations},$$ where $\|A\|=\max\{|a_{ij}|\}.$

			If ${\rm{SNF}}(A)=PAQ,$ then from \cite{newman}, it is proved that the last $n-r$ columns of $Q$ is a basis for the integer lattice $AX={\bf 0}.$ 

Below we include the proof of Theorem \ref{theorem}. Let $D=PAQ$ be the SNF of $A$, $k\times N$ $(k<N)$ matrix with ${\rm{rank}}(A)=k$  ($A$ is given in (\ref{equation:MATRIXA})). Also, 
    \[ Q=[{\bf q}_1\mid\cdots\mid {\bf q}_N].\]From properties of SNF, the $N-k$ vectors \({\bf q}_{k+1},\dots,{\bf q}_N$ is a basis of \(\ker_{\mathbb Z}(A)\).
Let also the $k$  vectors 
$\textbf{y}_{j} =q{\bf e}_{j},\ 1\le j\le k,$ where ${\bf {e}}_j$ the $j$th vector of the standard basis of ${\mathbb{R}}^N.$  We define
$V={\rm{span}}({\bf q}_{k+1},...,{\bf q}_{N})$ and $W={\rm{span}}({\bf e}_1,...,{\bf e}_k).$ 
\begin{theorem*}
If $V\cap W=\{ {\bf 0} \},$ then there are $N_1$, $N_2$ such that, the LLL-reduced matrix of $B_k$ is of the form $\hat{B_k}.$ In fact we prove that 
$2^{N+k}N_1^2<c(N,k)<N_2^2$ for some constant $c(N,k).$
\end{theorem*}
\begin{proof}  
 Let the matrix $\hat{B}_k=(\hat{b}_{i,j})$ be the LLL reduced matrix (row-wise) of $B.$ We shall prove that it has the form of (\ref{reduced_matrix}). It is enough to prove that $\hat{b}_{i,j}=0,$ for all $i\in\{0,...,N-1\}$ and $j\in\{N+1,...,N+k\}.$

We set ${\bf q}_j'=({\bf q}_j,{\bf 0}_{k+1})$ and ${\bf y}_j'=({\bf y}_j,{\bf 0}_{k+1})$ which belong to lattice ${\mathcal{L}}=L(B_k).$ Indeed, $A{\bf q}_j={\bf 0}$ and so $A{\bf q}_j={\bf 0}\pmod{q}$ and also $A{\bf y}_j={\bf 0}\pmod{q}.$
 From the hypothesis $V\cap W=\{{\bf 0}\},$
thus,
\[
\{\mathbf q_{k+1},\dots,\mathbf q_N\}
\;\cup\;
\{\mathbf y_1,\dots,\mathbf y_k\}
\]
is a basis of \(\mathbb R^{N}\). Moreover,
\[
\{\mathbf q_{k+1}',\dots,\mathbf q_N'\}
\;\cup\;
\{\mathbf y_1',\dots,\mathbf y_k'\}
\subset \mathcal{L}\]
is linear independent. We set,
$$c(N,k)=2^{N+k}\max\{ \|{\bf r}'\|^2,\| {\bf q}_{k+1}'\|^2,...,\|{\bf q}_{N}'\|^2 ,\|{\bf y}'_{1}\|^2,...,\|{\bf y}'_{k-1}\|^2\},$$ 
where ${\bf r}'=({\bf r},N_1,{\bf 0}_k)\in {\mathcal{L}}$ for some 
${\bf r}$ such that, $A{\bf r}^T={\bf T}_k\pmod{q}$\footnote{Since our system is of NTRU type there is always such a solution ${\bf r}.$}.

The set $\{{\bf r},{\bf q}_{k+1},\dots,{\bf q}_N,q{\bf e}_1,\dots,q{\bf e}_{k-1} \}$ ($N$ vectors) is linear independent in ${\mathbb{R}}^N.$ Indeed, it is enough to prove that 
 \[{\bf r}\not\in {\rm{span}}({\bf q}_{k+1},\dots,{\bf q}_{N},q{\bf e}_1,\dots,q{\bf e}_{k-1}).
 \]
 If ${\bf r}$ is a linear combination of the previous set, then by applying $A$ we get ${\bf T}_k={\bf 0}\pmod{q},$ which is a contradiction since ${\bf T}_k\not= {\bf 0}\pmod{q}.$ Therefore, from Lemma \ref{Lemma:lll-bound} we get
\begin{equation}\label{auxiliary1}
    \| \hat{{\bf b}}_j \|^2\le c(N,k), (0\le j\le N-1).
\end{equation} 
We choose,  
\begin{equation}\label{N1andN2}
    N_2^2>c(N,k).
\end{equation}

Suppose that $\hat{b}_{i_0,j_0}\not =0,$ for some $i_0 \{0,\dots,N-1\}$ and $j_0\in \{N+1,\dots,N+k \}.$ By the structure of $B_k$ we have $(B_k)_{i_0,j}\equiv 0 \pmod{N_2}$ for all
$j\in\{N+1,\dots,N+k+1\}$. Since LLL performs only unimodular integer row
operations, this congruence is preserved; hence
$(\widehat{B}_k)_{i_0,j}\equiv 0 \pmod{N_2}$ for the same indices $j$. 
So, $N_2|\hat{b}_{i_0,j_0},$ but $\hat{b}_{i_0,j_0}\not =0$, thus we get 
$|\hat{b}_{i_0,j_0}|\ge N_2.$ Therefore,
$$\| \hat{{\bf b}}_{i_0} \|^2\geq (\hat{b}_{i_0,j_0})^2\geq N_2^2>c(N,k),$$
which contradicts (\ref{auxiliary1}) if we set $j=i_0.$ 
We conclude that $\hat{b}_{i_0,j_0}=0,$ thus $\hat{b}_{i,j}=0$ for all $i\in\{0,...,N-1\}$ and $j\in\{N+1,...,N+k\}.$

Finally, $c(N,k)\ge 2^{N+k}\|{\bf r}' \|^2=2^{N+k}(N_1^2+\|{\bf r} \|^2)>2^{N+k}N_1^2.$ This shows that $2^{N+k}N_1^2<c(N,k)<N_2^2.$

\end{proof}
We have assumed that \(A\in {\mathbb{Z}}^{k\times N}\) has full row rank. For a random (full–row–rank) integer matrix \(A\), the condition
$
V\cap W=\{\mathbf 0\}
$
holds generically (i.e., for almost every choice of entries). Otherwise,
the homogeneous system \(AX=0\) over \(\mathbb R\) would have a nonzero solution of the form
\((a_1,\ldots,a_k,0,\ldots,0)\).
\end{appendix}


\begin{thebibliography}{9}

\bibitem{lenstra-system} K. Aardal, Cor A. J. Hurkens and A. K. Lenstra, Solving a System of Linear Diophantine Equations with Lower and Upper Bounds on the Variables, Integer Programming and Combinatorial Optimization. IPCO 1998, LNCS vol. {\bf 1412}, Springer 2000.

\bibitem{dsa2} M. Adamoudis, K. A. Draziotis and D. Poulakis, Attacking (EC)DSA scheme with ephemeral keys sharing specific bits. Theoretical Computer Science, Vol. {\bf 1001}, June 2024, Elsevier, 2024. 

\bibitem{marios} M. Adamoudis, K. A. Draziotis and D. Poulakis, Enhancing a DSA attack, CAI 2019, p. 13-25.  LNCS {\bf 11545}, Springer 2019.

\bibitem{adam_draz} M.~Adamoudis and K.~A.~Draziotis, Message recovery attack on NTRU using a lattice independent from the public key, Advances in Mathematics of Communications Volume 19(1), 2025, doi:\url{http://dx.doi.org/10.3934/amc.2023040}, arXiv:\url{https://arxiv.org/abs/2203.09620}

\bibitem{Albrecht} M.~Albrecht, S.~Bai, and L.~ Ducas, A subfield lattice attack on overstretched NTRU assumptions. CRYPTO 2016. LNCS \textbf{9814}, Springer 2016.

\bibitem{ntru-prime-submission-3} D. J. Bernstein, B. B. Brumley,  M. S. Chen, C. Chuengsatiansup, T. Lange, A. Marotzke,  B. Y. Peng, N. Tuveri, C. van Vredendaal, B. Y. Yang, NTRU - Prime Algorithm Specifications And Supporting Documentation, 2020.

\bibitem{gerald} G. Bourgeois and J. C. Faug\`ere, Algebraic attack on NTRU using Witt vectors and Gr\"obner bases, Journal of Mathematical Cryptology \textbf{3(3)} p. 205--214, 2009.

\bibitem{bradley}  G.~H.~Bradley, Algorithms for Hermite and Smith normal matrices and linear Diophantine equations, Math. Comput. {\bf 25}, American Mathematical Society (1971) p. 897--907.

\bibitem{CopSha97}  D. Coppersmith  and  A. Shamir,  Lattice  Attacks  on  NTRU.  In  Proc.  Eurocrypt  1997, LNCS \textbf{1223}, Springer, 2997.

\bibitem{nist2-ntru} C. Chen, O. Danba, J.  Hoffstein, A. H\"ulsing, J. Rijneveld, J. M. Schanck, T. Saito, P. Schwabe, W. Whyte, K. Xagawa, T. Yamakawa, and Z. Zhang, NTRU - Algorithm Specifications And Supporting Documentation, 2020.

\bibitem{Cheon} J. H. Cheon, J. Jeong, and C. Lee, An algorithm for NTRU problems and cryptanalysis of the GGH multilinear map without an encoding of zero. Cryptology ePrint Archive, Report \href{https://eprint.iacr.org/2016/139}{2016/139}, 2016.  

\bibitem{papadopoulou} K. A. Draziotis and A. Papadopoulou, Improved attacks on knapsack problem with their variants and a knapsack type ID-scheme, Advances in Mathematics of Communications, Volume 12, Issue 1, 2018, doi:\url{https://doi.org/10.3934/amc.2018026}.

\bibitem{ISIS23}  L. Ducas, T. Espitau and E. W. Postlethwaite, Finding Short Integer Solutions When the Modulus Is Small, Crypto 2023, LNCS {\bf 14083}, Springer, \url{https://eprint.iacr.org/2023/1125.pdf}

\bibitem{Gama-Nguyen} N. Gama and P. Q. Nguyen, New Chosen-Ciphertext Attacks on NTRU. Public Key Cryptography, PKC 2007, LNCS \textbf{4450}, Springer 2007.

\bibitem{np-complete}  M. R. Garey and D. S. Johnson, Computers and Intractability: A Guide to the Theory of NP-Completeness, Freeman, 1979.

\bibitem{gentry} C. Gentry, Key recovery and message attacks on NTRU-composite, EUROCRYPT 2001, LNCS \textbf{2045}, Springer 2001.

\bibitem{von} J.~von~zur~Gathen and M.~Sieveking, Weitere zum Erfiillungsproblem polynomial aquivalente kombinatorische Aufgaben, Komplexit\"at von Entscheidungsproblemen, pp:49-71, 1976.

\bibitem{hoffstein} J. Hoffstein, J. Pipher, and J. H. Silverman, NTRU: A ring-based public key cryptosystem, in Proceedings of ANTS '98 (ed. J. Buhler), LNCS {\bf 1423}, p. 267--288, 1998.

\bibitem{hybrid attack} N. Howgrave-Graham, A Hybrid Lattice-Reduction and Meet-in-the-Middle Attack Against NTRU. CRYPTO 2007, LNCS \textbf{4622}, Springer 2007. 

\bibitem{Howgrave} N. A. Howgrave-Graham and N. P. Smart, Lattice Attacks on Digital Signature Schemes,  {Des. Codes Cryptogr.} {\bf 23}, p. 283--290, 2001.

\bibitem{Odlyzko} N. Howgrave-Graham, J. H. Silverman, and W. Whyte, Meet-in-the-middle Attack on an NTRU private key, Technical report, NTRU Cryptosystems, July 2006. Report  04, available at \url{http://www.ntru.com}.

\bibitem{sidechannel} W.L. Huang, J.P. Chen and B.Y. Yang. Power Analysis on NTRU Prime. IACR Transactions on Cryptographic Hardware and Embedded Systems, Vol. {\bf 2020(1)}, DOI: \url{https://doi.org/10.13154/tches.v2020.i1.123-151}

\bibitem{sidechannel2} E. Karabulut, E. Alkim and A. Aysu, Single-Trace Side-Channel Attacks on $\omega$-Small Polynomial Sampling: With Applications to NTRU, NTRU Prime, and CRYSTALS-DILITHIUM, 2021 IEEE International Symposium on Hardware Oriented Security and Trust (HOST), Tysons Corner, VA, USA, 2021, pp. 35--45, doi: \url{https://doi.org/10.1109/HOST49136.2021.9702284}.

\bibitem{Kirchner} P. Kirchner and P. A. Fouque, Revisiting Lattice Attacks on Overstretched NTRU Parameters. Eurocrypt 2017, LNCS {\textbf{10210}}, Springer 2017.

\bibitem{improvedlll} P. Kirchner and T. Espitau and P. A. Fouque, An Improved LLL Algorithm, Advances in Cryptology -- ASIACRYPT 2019, Springer, 2019

\bibitem{may_recent} E. Kirshanova, A. May, and J. Nowakowski, New NTRU Records with Improved Lattice Bases. PQCrypto 2023, LNCS {\bf 14154},  2023, doi : \url{https://doi.org/10.1007/978-3-031-40003-2_7}

\bibitem{branch-and-bound} R. M. Kolpakov and M. A. Posypkin, Upper and lower bounds for the complexity of the branch and bound method for the knapsack problem, Discrete Mathematics and Applications, 2010. 

\bibitem{segment} H. Koy and C. P. Schnorr, Segment Lattice Reduction, Proceedings of the Workshop on the Theory and Application of Cryptographic Techniques, EUROCRYPT '89, Springer, 1990

\bibitem{lll} A. K. Lenstra, H. W. Lenstra, L. Lov\'asz, Factoring polynomials with rational coefficients,  Math. Ann. {\bf 261}, 515--534 (1982). \url{https://doi.org/10.1007/BF01457454}

\bibitem{may-hints} A. May and J. Nowakowski, Too Many Hints -- When LLL Breaks LWE, 2024, \url{https://eprint.iacr.org/2023/777.pdf}

\bibitem{floatinglll} A. Neumaier and D. Stehl{\'{e}}, Floating-Point {LLL} Revisited, Advances in Cryptology -- EUROCRYPT 2016, Springer, 2016

\bibitem{newman}  M.~Newman, The Smith normal form. Proceedings of the Fifth Conference of the International Linear Algebra Society, Linear Algebra Appl. {\bf 254}, p. 367--381, Elsevier 1997.

\bibitem{Nguyen - Boosting the hybrid}  P. Q. Nguyen, Boosting the Hybrid Attack on NTRU: Torus LSH, Permuted HNF and Boxed Sphere, Third PQC Standardization Conference, 2021.

\bibitem{ntru_vfk_attack} E. Poimenidou, M. Adamoudis, K. A. Draziotis, and K. Tsichlas, Message Recovery Attack in NTRU through VFK Lattices, preprint, \url{https://doi.org/10.48550/arXiv.2311.17022}

\bibitem{nutmic} E. Poimenidou, M. Adamoudis, K. A. Draziotis, Towards message recovery in NTRU Encryption with auxiliary data, NuTMiC 2024, LNCS {\bf 14966}, Springer. 

\bibitem{flatter} K. Ryan and N. Heninger, Fast Practical Lattice Reduction Through Iterated Compression, Advances in Cryptology -- CRYPTO 2023, LNCS {\bf 14083}, Springer.   
\bibitem{Smith} H.~J.~S.~Smith, On systems of linear indeterminate equations and congruences. Phil. Trans. Roy. Soc. London {\bf 151}, p. 293--326, 1861.

\bibitem{storj} A.~Storjohann. Computing hermite and smith normal forms of triangular integer matrices. Linear Algebra and its Applications {\bf 282} p.25--45, Elsevier, 1998.

\bibitem{shor} P. W. Shor, Algorithms for quantum computation: Discrete logarithms and factoring. In 35th Annual Symposium on Foundations of Computer Science, Santa Fe, New Mexico, USA, 20-22 November 1994, p. 124--134. IEEE Computer Society, 1994.


\bibitem{witt} J. H. Silverman, N. P. Smart, and F. Vercauteren, An algebraic approach to NTRU ($q = 2n$) via Witt vectors and overdetermined systems of non linear equations. Security in Communication Networks -- SCN 2004, LNCS {\bf 3352}, p. 278--298. Springer, 2005.

\end{thebibliography}
\end{document}